\documentclass[a4paper]{article}

\usepackage[USenglish]{babel}
\usepackage{authblk}
\usepackage{lmodern,fullpage}
\usepackage[colorlinks = true, linkcolor = blue, urlcolor  = blue, citecolor = blue, anchorcolor = blue]{hyperref}
\usepackage{amsmath}
\usepackage{amsthm}
\usepackage{amssymb}
\usepackage{amsfonts}
\usepackage{mathrsfs}
\usepackage{mathtools}
\usepackage{cmap} 
\usepackage{outlines}
\usepackage{cite}

\newcommand{\eps}{\varepsilon}

\newcommand{\oeps}{1+\eps}
\newcommand{\gs}{\mathcal{G}}
\newcommand{\psbp}{$\mathsf{PSBP}$}
\newcommand{\dsbp}{$\mathsf{DSBP}$}
\newcommand{\bsz}{\mathcal{Z}} 
\newcommand{\bns}{\mathcal{B}} 
\newtheorem{lemma}{Lemma}
\newtheorem{theorem}{Theorem}

\title{EPTAS for the Dual of Splittable Bin Packing with Cardinality Constraint\thanks{Supported in part by ISF - Israeli Science Foundation grant number 308/18.}}
\author[1]{G. Jaykrishnan} 
\author[2]{Asaf Levin}
\affil[1]{Faculty of Industrial Engineering and Management, The Technion, Haifa, Israel.\ \texttt{jaykrishnang@hotmail.com}}
\affil[2]{Faculty of Industrial Engineering and Management, The Technion, Haifa, Israel.\ \texttt{levinas@ie.technion.ac.il}}
\date{}

\begin{document}
\maketitle

\begin{abstract}
The problem considered is the splittable bin packing with cardinality constraint. It is a variant of the bin packing problem where items are allowed to be split into parts but the number of parts in each bin is at most a given upper bound. Two versions of the splittable bin packing with cardinality constraint have been studied in the literature. Among these variants we consider the dual one where the objective is to minimize the maximum bin size while packing (may be fractional) the items to a given set of bins. We exhibit an EPTAS for the dual problem when the cardinality upper bound is part of the input. This result answers an open question raised by Epstein, Levin, and van Stee \cite{epstein2012}.
\end{abstract}

\section{Introduction}

The splittable bin packing with cardinality constraint problem was suggested as a compromise between memory allocation and parallelism. In that setting, the problem allows memory sharing while limiting the number of ports i.e., the processors that can access the shared memory.
In this problem we have a set of items $T$. Each item $t$ is associated with its {\em size} $S_t\geq 0$. The item sizes are not bounded (from above). 

A \emph{fractional packing} of items to bins is a packing wherein the items are allowed to be split into parts (without any restriction on the sizes or the number of parts) and the parts can be packed to any bin without restriction except for the bins size. An optimal solution to bin packing, when fractional packing is allowed, can be obtained in polynomial time as follows. Pack items sequentially (in an arbitrary order) to bins and split items when the total size of parts in the current bin is about to exceed the bin size. In this optimal solution there is at most one partially filled bin. Enforcing constraints can make the problem hard. There are many constraints considered in literature. Limiting the number of parts generated from an item, limiting the part sizes, and item conflicts are some constraints considered in the literature. We consider the \emph{cardinality constraint} which defines an upper bound $k$ on the number of parts of different items in each bin. Enforcing this constraint makes the problem NP-hard. Two versions of the splittable bin packing with cardinality constraint have been studied in the literature. 

The primal problem, denoted as \psbp{}, is defined as follows. Given a bin size $\bsz>0$, a feasible solution  is a packing of all items to bins (fractional packing is allowed) such that the total size of parts of items in each bin is at most $\bsz$ and the number of parts of different items in each bin is at most $k$. The goal function is to minimize the total number of used bins. 

The dual problem, named the {\sc Dual Splittable Bin Packing with cardinality constraint} is denoted as \dsbp{}. It is defined as follows. Given a set of bins $\bns$, a feasible solution  is a packing of all items (fractional packing is allowed) to the bins in $\bns$ such that the number of parts of different items in each bin is at most $k$. The goal function is to minimize the maximum bin size.  Observe that in this \dsbp\ the input consists of the set of items $T$, each of which associated with a size, the set of bins $\bns$, and the value of $k$.

In order to present the known results for these problems, we recall the definitions of approximation schemes both in the absolute sense and in the asymptotic sense.  Our problems are minimization problems so the objective function is referred to as cost. For an algorithm $\mathcal{A}$, we denote its cost by $\mathcal{A}$ as well. The cost of an optimal algorithm is denoted by $\mathsf{OPT}$. We define the \emph{asymptotic approximation ratio} of an algorithm $\mathcal{A}$ as the infimum $\mathcal{R} \geq 1$ such that for any input, $\mathcal{A} \leq \mathcal{R}\cdot\mathsf{OPT}+c$, where $c$ is independent of the input. If we enforce $c = 0$, $\mathcal{R}$ is called the \emph{absolute approximation ratio}. An \emph{asymptotic polynomial time approximation scheme} is a family of approximation algorithms such that for every $\eps \geq 0$, the family contains a polynomial time algorithm with an asymptotic approximation ratio of $\oeps$. We abbreviate asymptotic polynomial time approximation scheme by APTAS (also called an asymptotic PTAS). An \emph{asymptotic fully polynomial time approximation scheme} (AFPTAS) is an APTAS whose time complexity is polynomial not only in the input size but also in $1/\eps$. If the scheme satisfies the definition above with $c = 0$, stronger results are obtained, namely, polynomial time approximation schemes and fully polynomial time approximation schemes, which are abbreviated as PTASs and FPTASs, respectively. An \emph{efficient polynomial time approximation scheme} (EPTAS) is a PTAS whose time complexity is upper bounded by the form $f(1/\eps)\cdot poly(n)$ where $f$ is some computable (not necessarily polynomial) function and $poly(n)$ is a polynomial of the length of the (binary) encoding of the input. The notion of EPTAS is a modern one, motivated in the fixed parameterized tractable (FPT) community (see e.g. \cite{Cesati97}).

\paragraph{Literature review.}
\psbp{} was first considered in \cite{chung2006} with $k=2$. They show that the problem is NP-hard in the strong sense and also gave a $3/2$ approximation algorithm for this case. \cite{epstein2011improved} further generalized the primal problem for fixed values of $k\geq 3$. The authors show that the problem is NP-hard for every fixed value of $k$ and provide a $(7/5 + \eps)$-approximation algorithm for $k=2$ and for any $\eps\geq 0$. In \cite{epstein2012} special cases of the  primal problem are solved in polynomial time (polynomial in both $|T|$ and the binary encoding length of the input). For the general case when the problem becomes NP-hard, they present an EPTAS. Furthermore, the problem is strongly NP-hard so this is the best achievable scheme since an FPTAS would imply P=NP. Next, we note that the existence of an AFPTAS for the primal problem would also imply P=NP. The case when the problem becomes NP-hard is when the cost of the optimal solution is very large. In such cases the asymptotic approximation ratio and the absolute approximation ratio becomes the same up to an additive $\eps$ in the value of approximation ratio.  Therefore,  an AFPTAS would result in an FPTAS which is impossible unless P=NP. Thus, an AFPTAS is not achievable for the primal problem. 

Most relevant to our work, \cite{epstein2012} also presents a PTAS for \dsbp\ with constant values of $k$. The authors leave the existence of an EPTAS for the dual problem as an open problem. The assumption that $k$ is a constant means that in their scheme the time complexity is exponential in $k$. Our result is an EPTAS for \dsbp{} that need not assume that $k$ is a constant. Namely, we both generalize the cases for which there exist an approximation scheme for \dsbp{}, and improve the time complexity of the scheme for the cases in which there was a previous PTAS.

\paragraph{Related problems.}
Many variations of the splittable bin packing with cardinality constraint have been studied.  See \cite{ekici2021bin,castro2019study,casazza2016exactly} for some examples of exact (exponential time) algorithms and heuristics for such variants. Here, we would like to mention the bin packing with item fragmentation problem.  It is another variant of the bin packing problem  wherein the items to be packed can be split. See \cite{casazza2016exactly,bertazzi2019bin,ekici2021bin} for studies of problems with item fragmentation. The main difference between bin packing with item fragmentation and our problem is that bin packing with item fragmentation usually restricts the number of parts (of an item or total) while our problem restricts the number of parts of items that can be packed into a bin without any restriction on the number of splits of an item.

Cardinality constrained scheduling is the problem of machine scheduling on identical machines wherein there is an upper bound on the number of jobs that can be assigned to a machine and the goal is to minimize the makespan. This problem is in fact a special case of \dsbp{} as we can assume that $|T| = k \cdot |\bns|$ by adding items of size $0$, and then it is impossible to split an item. This problem admits an EPTAS \cite{chen2016}. For more on this problem please see \cite{dell2001bounds, dell2006lower, he2003kappa, chen2016, epstein2022cardinality}.

By removing the cardinality constraint in the above cardinality constrained scheduling problem, we obtain the setting of makespan minimization on identical parallel machines. We refer to \cite{Hochbaum87,hochbaum97,Alon98} for an EPTAS for this setting. \cite{JKV16} improved the time complexity of these schemes. 

Bin packing with cardinality constraint (without splitting the items) was studied in \cite{krause1975analysis,krause1977errata,kellerer1999cardinality,caprara2003approximation,babel2004algorithms,epstein2010, jansen2019approximation,balogh2020online}.
Bin packing is the problem without cardinality constraint (without splitting the items), see e.g.  \cite{fernandez1981bin,karmarkar1982efficient}. 

\paragraph{Notation.}
Given a positive integer $N$, we let $[N] = \{1,2,\ldots,N\}$ be the set of natural numbers up to $N$. We let $\mathbb{R}_{+}$ denote the set of non-negative reals and $\mathbb{R}_{++}$ the set of strictly positive reals. We use similar notation of $\mathbb{Z}_{+}$ and $\mathbb{Z}_{++}$ for integers. Let $\eps\in\mathbb{R}_{++}$ such that $1/\eps \in\mathbb{Z}:1/\eps \geq 10$. The total size of all items is denoted by $W$, that is, $W=\sum_{t\in T} S_t$.

\paragraph{Paper outline.}
For the remainder of the paper we are dealing with \dsbp{}. Recall that the number of bins determined in the input is denoted by $|\bns|$. This problem can be interpreted as fractional scheduling on parallel identical machines with the goal of minimizing the makespan but the problem is non-trivial as we have an upper bound of $k$ on the number of parts of different jobs that can be placed on a machine.

Our scheme starts by imposing a structure on the input (Section \ref{pre}). We use a guessing step, to guess approximated information on the optimal solution. For each value of the guess, we perform the following. We identify a near optimal solution with a particular structure we detail in Section \ref{structure}. This structure is exploited to generate partial information on the solution.  This partial information  is used in a formulation of a mixed integer linear program (MILP)  (see Section \ref{milp}) directing our algorithm.  The next step of the scheme is to obtain an optimal solution of the MILP. The best solution among all the optimal solutions to all possible values of the guess is converted to a feasible packing for \dsbp\ whose cost is approximately the cost of the corresponding optimal solution to the MILP (Section \ref{milp_convert}).

\section{Structuring the input}\label{pre}
Our scheme starts by guessing the optimal maximum bin size and assume that this bin size is $1$, then we perform a rounding step, and then we classify the items into small and large. All these steps can be considered as part of structuring the input and these are the steps we consider in this section. Later, we provide the remaining steps of the algorithm. 

\subsection{Guessing the optimal bin size}\label{guessing}
The optimal bin size is in the interval $\in [W/|\bns|, W]$. We guess the approximated upper bound on bin size as an integer power of $\oeps$. We guess a value $g$ that is an integer power of $\oeps$ and that the maximum bin size of an optimal solution to the problem is in the interval $(g/(1+\eps),g]$.  We conclude the following lemma.
\begin{lemma}
	The number of possible values for the guess is at most $\log_{(\oeps)} |\bns|+ 2$.
\end{lemma}
If the guess is $g$, we scale down the item sizes of all items in the original instance by dividing the sizes of all items by $g$ and initialize the approximated upper bound on the size of bins to $1$. This scaled down instance is denoted by $I$. 

\subsection{Rounding step}\label{rounding}
We next convert the scaled down instance $I$ into a new instance $I'$ for which the following two properties hold. 
\begin{outline}[enumerate]
	\1 The maximum size of an item in $I'$ is at most $1/\eps^2$. \label{one}
	\1 For every item whose size in $I'$ is at least $\eps$, the size of the item is an integer multiple of $\eps^2$.\label{two}
\end{outline}

Property \ref{one} is obtained by splitting items of size strictly greater than $1/\eps^2$. Every item $x$ of size $X$ strictly greater than $1/\eps^2$, is replaced with $\lfloor \eps^2 X\rfloor$ items of size $1/\eps^2$ and one additional item of size $X - \lfloor \eps^2 X\rfloor\cdot(1/\eps^2)$. Property \ref{two} is achieved by rounding up sizes of items at least $\eps$ to integer multiples of $\eps^2$. We denote by $I'$ the instance obtained after the preceding transformations.

\begin{lemma}
A packing of $I$ can be converted to a packing of $I'$ by increasing the size of bin by at most $3\eps$. A packing of $I'$ can be converted to a packing of $I$ in polynomial time without increasing the size of bins. 
\end{lemma}
\begin{proof}
Consider a packing of instance $I$. Modify each item $x$ of size $X> 1/\eps^2$ as follows. First increase the size of parts of $x$ in each bin by a multiplicative factor of $\oeps$, adding empty space as necessary. Arbitrarily order the bins in which parts of $x$ are packed and denote the ordered list of those bins by $B_x$. Let $S_b, \forall b\in B_x$ be the total size of parts of $x$ in bins up to $b$. For each $i=1,2,\ldots,\lfloor\eps^2 X\rfloor$ identify the bin $b_i\in B_x$ such that $S_{b_i-1} < i/\eps^2$ and $S_{b_i} \geq i/\eps^2$. Cut the part of $x$ in bin $b_i$ into two parts of size $i/\eps^2-S_{b_i-1}$ and $S_{b_i}-i/\eps^2$, thus increasing the number of parts in $b_i$. Notice that the empty space added in bins $b_{(i-1)}, \ldots, b_i - 1$ is $\eps\cdot 1/\eps^2 = 1/\eps$. Move the first part of $x$ in $b_i$ to the bins between $b_{(i-1)}$ and $b_{i}-1$ arbitrarily. There is enough empty space available to do this since the size of the part is at most $1<1/\eps$. The increase in the size of a bin $b$ after performing the above conversion to all items of size strictly larger than $1/\eps^2$ is at most $\eps \mathsf{S}_b\leq \eps$, where $\mathsf{S}_{b}$ is the total size of parts of items larger than $1/\eps^2$ packed into bin $b$ in the packing of $I$. Let this intermediate instance, after the above conversion, be $\underline{I}$. 

Next, perform the following conversion for each item $y$ of $\underline{I}$ of size $Y\geq\eps$. Let $B_y$ be the set of the bins in which parts of $y$ are packed. Increase the size of all parts of $y$ by a multiplicative factor of $\oeps$, adding empty space as necessary. Then, we add a part of $y$ of size $Y'$ such that $Y'=\eps^2\lceil Y/\eps^2\rceil- Y < \eps^2$. This part is added arbitrarily to the bins in $B_y$ by increasing the sizes of the parts of $y$ in the bins accordingly. The number of parts of items in each bin is not increased. This is a feasible packing since the empty space created is of size at least $\eps^2$ using $Y\geq \eps$. This increases the space of a bin $b$ by at most $\eps(\oeps) \leq 2\eps$. This gives the required packing of instance $I'$. 

For the other direction consider a packing of instance $I'$. Decrease the size of items of size strictly greater than $\eps^2$ to the original size as in $\underline{I}$ in an arbitrary way. This does not increase the maximum size of a bin or the number of parts in each bin. Now we have a packing of instance $\underline{I}$. Any packing of instance $\underline{I}$ is a packing of instance $I$ since the items of $I$ can be split arbitrarily and this does not increase the number of parts in each bin with respect to the packing of $\underline{I}$. Thus, we get a packing of $I$ without increasing the maximum bin size.
\end{proof}

After the above rounding we have that the approximated upper bound on the bin sizes is denoted as $\gs'$ and we initialize its value to be $\gs'= 1+3\eps$. Based on the guessing step, we are guaranteed there exists a feasible solution to the instance $I'$ whose maximum bin size is at most $\gs'$.

\subsection{Classification of items}
The items are classified into two categories. An item is called \emph{small} if the size of the item is strictly smaller than $\eps$, else it is \emph{large}. The size of a large item is in the interval $[\eps,1/\eps^2]$ and is an integer multiple of $\eps^2$ (as a result of the rounding step). Let $\mathsf{L}$ denote the set of distinct sizes of large items in $I'$. We have, $|\mathsf{L}| \leq 1/\eps^4-1/\eps+1 \leq 1/\eps^4$ and the upper bound is a constant when $\eps$ is fixed. Let $\mathcal{L}$ denote the set of all large items, $\mathcal{L}_\ell$ the set of large items of size $\ell\in\mathsf{L}$, and $\mathcal{S}$ the set of small items.

\section{Nice packing}\label{structure}
We next show that every feasible packing can be converted to another packing where items of size at least $\eps$ are only cut at integer multiples of $\eps^2$ when packed. To prove this we use Lemma \ref{lem3} quoted below that appears in \cite[Theorem 2.1]{lenstra1990} (see also \cite{shmoys1993}). We will use this property to assist our scheme in obtaining a near optimal solution.

In order to present the settings of \cite{lenstra1990, shmoys1993} consider the scenario where we have $n$ items and $m$ (non-identical) bins. The size of item $i$ in bin $j$ is denoted by $s_{ij}$, $G_j$ is a given upper bound on the size of bin $j$, and $g\in\mathbb{R}_+$. The meaning of the decision variables is that if item $i$ is assigned completely to bin $j$, then $x_{ij}=1$. Define the following feasibility linear program denoted as \ref{lp} which gives a fractional packing of items to bins.

\begin{align*}\tag{LP(g)}\label{lp}
	\sum_{j=1}^{m} x_{ij} &= 1,\ i = 1,2,\ldots,n,\\
	\sum_{i=1}^{n} s_{ij}x_{ij} &\leq G_j,\ j=1,2,\ldots,m\\
	x_{ij} &\geq 0,\ i=1,2,\ldots,n,\ j=1,2,\ldots,m\\
	x_{ij} &= 0 \text{ if } s_{ij}>g,\ i=1,2,\ldots,n\ ,j=1,2,\ldots,m.
\end{align*}

\begin{lemma}\label{lem3}\cite{lenstra1990, shmoys1993}
	If \ref{lp} has a feasible solution $\textbf{x}$, then there exists a feasible integer packing $\textbf{x}'$ where the size of bin $j$ is at most $G_j+g, j=1,2,\ldots,m$. Furthermore, in $\textbf{x}'$ if an item $i$ is packed in a bin $j$, then $x_{ij}>0$.\label{paper}
\end{lemma}	
	In our scenario, the second part of the lemma will imply that the number of parts of items in each bin is not increased in $\textbf{x}'$ compared to $\textbf{x}$.

\begin{lemma}\label{nice}
Every feasible packing can be converted to another feasible packing with maximum bin size at most $1+4\eps$ where every part of every large item has a size that is an integer multiple of $\eps^2$.
\end{lemma}
\begin{proof}
Consider a feasible packing $P$ of instance $I'$ in which items are split arbitrarily. We use Lemma \ref{paper} to transform $P$ to a packing that has the required property. Let the size of bin $j$ in $P$ be $Z_j$ for all $j\in \bns$. 

We split the large items as follows. For each $\ell\in\mathsf{L}$, each large item of size $\ell$ is replaced by $\ell/\eps^2$ items of size $\eps^2$. Items thus created are called \textit{auxiliary items}. Each large item can create at most $1/\eps^4$ auxiliary items. Let this intermediate instance be $I''$ and the set of all auxiliary items in $I''$ be $\mathcal{I}$. Notice that $P$ is a packing of $I''$ by replacing the large item by the corresponding auxiliary items it generates, even though it may be infeasible for \dsbp{} for $I''$ (as the number of parts in a bin may increase). The packing of small items in each bin is left as in $P$. We modify the packing of the auxiliary items generated to get the structure we need. Thus, the effective space available for auxiliary items in a bin $j$ is now $G_j = Z_j - S_j$ where $S_j$ is the total size of parts of small items in $j$.

Next, we generate $\textbf{x}$ values for all $i\in\mathcal{I}$ and $j\in\bns$ as follows. If a portion of size $\tau \eps^2, \tau\in[0,1]$ of $i\in \mathcal{I}$ is assigned to bin $j$, then $x_{ij}=\tau$. We modify the item sizes in preparation to apply Lemma \ref{paper}. The size of part $i\in \mathcal{I}$ in bin $j$ is denoted by $s_{ij}$. If $x_{ij}>0$, then $s_{ij}=\eps^2$, and otherwise $s_{ij}=+\infty$. The $\textbf{x}$ values now satisfy \ref{lp} with $g=\eps^2$. Using Lemma \ref{paper} we convert $\textbf{x}$ to an integral packing of auxiliary items, denoted by $\textbf{x}''$. The total size of auxiliary items in bin $j$ with respect to $\textbf{x}''$ is  at most $G_j + \eps^2$. Considering the preexisting small items in bin $j$, the size of $j$ is at most $G_j + \eps^2 + S_j = Z_j + \eps^2 \leq \gs'+\eps^2 \leq (1+3\eps)+\eps = 1+4\eps$. 

Next, we modify the $\textbf{x}''$ values for auxiliary items in $I''$ to $\textbf{x}'$ for only large items of $I$. The ${x}'_{ij}$ value of a large item $i$ in bin $j$ is $\sum_{i' \in \mathcal{I}_i} x''_{i'j}$ where $\mathcal{I}_i$ is the set of auxiliary items generated from large item $i$. This is done for each large item and bin to get $\textbf{x}'$. That is, if there is a bin with multiple auxiliary items from a single large item, we consider those auxiliary items as one part of that item. The size of a part of a large item $i$ in a bin $j$ is $x'_{ij} \eps^2$. Since $\textbf{x}''$ was integer, $x'_{ij} = \sum_{i' \in \mathcal{I}_i} x''_{i'j} \in \mathbb{Z}$. Thus, the size of all parts of all large items in every bin is an integer multiple of $\eps^2$. Furthermore, in $\textbf{x}'$ the number of parts of different items in $I'$ in each bin is not more than that in $\textbf{x}$ and thus $\textbf{x}'$ is a feasible packing to \dsbp{}.
\end{proof}

A packing that satisfies Lemma \ref{nice} is called a \emph{nice packing}. Next, we redefine $\gs'=1+4\eps$ and note that it is an integer multiple of $\eps^2$. By the above steps we conclude that there exists a nice packing of $I'$ with maximum bin size at most $\gs'$, and in order to establish an EPTAS it suffices to construct a feasible solution with maximum bin size at most $(1+2\eps)\gs'$ in the required time complexity. 

\section{MILP}\label{milp}
Patterns and configurations  define partial information on the solution and are used in the MILP. The MILP  finds the patterns by which the large items are split. The MILP also finds the configurations according to which the parts generated from large items (based on the patterns) are packed in bins. In addition to these tasks, the MILP also provides a fractional packing of small items to the configurations. The optimal solution to the MILP is found in polynomial time using \cite{lenstra1983,kannan1983} since the number of integer variables is a constant. The MILP solution can be transformed to a feasible packing with a similar cost as we demonstrate in the proof of Theorem \ref{sol2pack}.

\subsection{Pattern}
A pattern is a vector of length $\gs'/\eps^2 + 1$ and each large item is assigned to a pattern which represents how the item is split. The components of a pattern $p$ are as follows.
\begin{outline}
	\1 The first component (denoted by $\alpha_p\in\mathbb{Z}_+$) stores the size of the large item as an integer multiple of $\eps^2$ to which this pattern is assigned. That is, $\alpha_p\eps^2$ is the size of the large item to which $p$ is assigned. 
	\1 The next $r\in[\gs'/\eps^2]$ components (denoted by $\beta_{pr}\in\mathbb{Z}_+$)  determine how many parts of size $r\eps^2$ (for all $r$) are cut from a large item assigned pattern $p$.
\end{outline} 

We say that a pattern $p$ is a {\em feasible pattern} if 
\begin{align}\label{feasible_pattern}
	\alpha_{p} = \sum_{r\in[\gs'/\eps^2]}\beta_{pr}r.
\end{align} 
Let $\mathcal{P}$ be the set of feasible patterns. 

\begin{lemma}
	The number of feasible patterns is at most $\left(1/\eps^4 + 1\right)^{(2/\eps^2)}$. 
\end{lemma}
\begin{proof}
	The $\alpha_p$ value of a feasible pattern is well-defined given the $\beta_{pr}$ for all $r$. Each of the $\beta_{pr}$ components can have at most $\left(1/\eps^4 + 1\right)$ values (by \eqref{feasible_pattern}) since $\alpha_p \leq 1/\eps^4$. The claim follow from $\gs'/\eps^2 < 2/\eps^2 - 1$ since $1/\eps\geq 10$.
\end{proof}

\subsection{Configuration}
Every bin is assigned a configuration and this configuration represents the packing of the bin. The components of a configuration $c$ are as follows.
\begin{outline}
	\1 The first component (denoted by $\gamma_c\in\mathbb{Z}_+$) represents the total size of parts of small items assigned to a bin (assigned this configuration) rounded down to an integer multiple of $\eps^2$ such that the size of the small items in such a bin is in the interval $[\gamma_c\eps^2, (\gamma_c+1)\eps^2)$.

	\1 The next $r\in[\gs'/\eps^2]$ components (denoted by $\delta_{cr}\in\mathbb{Z}_+$) store the number of different parts of a particular size (as an integer multiple of $\eps^2$) cut from large items packed to a bin assigned this configuration. For each $r\in[\gs'/\eps^2]$, $\delta_{cr}$ states the number of different parts of size $r\eps^2$ (cut from large items) packed to a bin assigned configuration $c$.
\end{outline}

\subsubsection{Creating a partial packing}
To create a {\em partial packing} of a bin from a configuration $c$ do the following. First assign a single virtual item of size $\gamma_c \eps^2$. A virtual item is not an item of instance $I'$ and it is used as a placeholder for parts of small items. Next, assign $\delta_{cr}$ parts of size $r \eps^2$ (generated from large items) in some order of $r\in[\gs'/\eps^2]$ to the bin. A configuration $c$ is a {\em feasible configuration} if the size of the bin in the partial packing created from $c$ is at most $\gs'$. Let $\mathcal{C}$ be the set of all feasible configurations.

\begin{lemma}\label{feasible_configuration}
	The number of feasible configurations is at most $(2/\eps^2)^{(2/\eps^2)}$. 
\end{lemma}
\begin{proof}
	The possible values for the first component of the configuration is at most $\gs'/\eps^2 + 1$. Each of the next $\gs'/\eps^2$ components can have at most $\gs'/\eps^2+1$ values since the maximum number of parts (each of which is of size $\eps^2$) that can be assigned to a bin of size $\gs'$ is at most $\gs'/\eps^2$. The upper bound follows from $\gs'/\eps^2 + 1\leq 2/\eps^2$ since $1/\eps\geq 10$.
\end{proof}

\subsection{The families of decision variables}
	\paragraph{Assignment variables for small items to configurations.} We have an assignment variable $x_{ic}$ for small item $i$ and configuration $c$. If small item $i$ is assigned completely to bins with configuration $c$, then $x_{ic}=1$. There are at most $n\cdot |\mathcal{C}|$ such variables. These variables are allowed to be fractional.

	\paragraph{Configuration counters.}
	There is a configuration counter for every feasible configuration. For a configuration $c\in \mathcal{C}$, the variable $y_c$ denotes the number of bins to which configuration $c$ is assigned. There are at most $|\mathcal{C}|$ variables and these variables are forced to be integral.

	\paragraph{Pattern counter.}
	There is a pattern counter for every feasible pattern. For a pattern $p \in \mathcal{P}$, the variable $z_p$ denotes the number of large items to which pattern $p$ is assigned. There are $|\mathcal{P}|$ such variables and these variables are forced to be integral.

	\paragraph{Parts counter.}
	For every distinct size of parts, as an integer multiple of $\eps^2$, we have a counter. For $r\in[\gs'/\eps^2]$, $v_r$ states the number of different parts (cut from large items) of size $r\eps^2$ available for packing. These variables are allowed to be fractional.

\subsection{The constraints}

Next, we present the constraints of the MILP with their intuitive meaning.  This intuition is not used to prove our claims but it is presented to ease the reading of the constraints.

Each bin should be assigned a configuration and since one configuration can be assigned to multiple bins, we have
		\begin{align}\label{con_bin}
			\sum_{c\in\mathcal{C}}y_c = |\bns|.
		\end{align}
	
Each small item $i$ should be packed completely. Since such a small item can be split between configurations, we have
		\begin{align}\label{sma_con}
			\sum_{c\in\mathcal{C}}x_{ic} = 1,\ \forall i\in\mathcal{S} .
		\end{align}
	
All the parts that are generated from large items should be packed.
		\begin{align}
			v_r &= \sum_{p\in \mathcal{P}}\beta_{pr}z_p,\ \forall r\in[\gs'/\eps^2].\label{par_cou}\\
			v_r &= \sum_{c\in \mathcal{C}} \delta_{cr}y_c,\ \forall r\in[\gs'/\eps^2].\label{con_par}
		\end{align}
	The former constraint calculates the number of parts of each size generated from large items based on the patterns. Observe that by this constraint and the fact that $\textbf{z}$ values are forced to be integral, we conclude that in every feasible solution the $\textbf{v}$ values are also integers. The constraint \eqref{con_par} chooses the configurations such that there are enough positions to pack these parts. 

Each large item should be assigned a pattern and since some large items can be assigned to a common pattern, we have
	\begin{align}\label{pat_lar}
		\sum_{p\in \mathcal{P}:\alpha_p = \ell/\eps^2} z_p= |\mathcal{L}_\ell|,\ \forall \ell\in \mathsf{L}.
	\end{align}

	Each configuration should satisfy the cardinality constraint since each configuration represents a bin. The total number of parts of large items that can be assigned to a configuration is determined by the configuration and the number of copies of this configuration. Thus, the average number of parts of small items that can be assigned to one copy of a configuration is at most the difference between $k$ and the total number of parts of large items allowed by the configuration. We get, 
		\begin{align}\label{card}
			\sum_{i\in \mathcal{S}} x_{ic} \leq y_c \left(k - \sum_{r\in[\gs'/\eps^2]}\delta_{cr}\right),\ \forall c\in \mathcal{C}.
		\end{align}
	
The total size of parts of small items assigned to a configuration should be at most the size allowed by each configuration. Since the total size of parts of small items allowed was rounded down in each configuration, we increase the space for each copy of the configuration by $\eps^2$. That is,
		\begin{align}\label{tot_sma}
			\sum_{i\in \mathcal{S}}x_{ic}S_i \leq y_c(\gamma_c+1)\eps^2,\ \forall c \in \mathcal{C}.
		\end{align}
	Lastly, we have the non-negativity of all variables and integrality constraints of some of the variables.
		\begin{align}
			v_{r}&\in\mathbb{R}_+,\ \forall r\in [\gs'/\eps^2].\label{non_v}\\
			x_{ic}&\in\mathbb{R}_+,\ \forall c\in \mathcal{C}, \forall i\in \mathcal{S}.\label{non_x}\\ 
			y_c&\in\mathbb{Z}_+,\ \forall c\in \mathcal{C}.\label{non_y}\\ 
			z_p&\in\mathbb{Z}_+,\ \forall p\in \mathcal{P}.\label{non_z}
		\end{align}

This concludes the presentation of the MILP.  Next, we turn our attention to the analysis of this MILP.

\begin{theorem}\label{milp_analysis}
If there exists a feasible nice packing of the instance $I'$ with maximum bin size $\gs'$, then the MILP has a feasible solution.
\end{theorem}		
\begin{proof}
Consider a nice packing $P'$ for the instance $I'$ with maximum bin size $\gs'$ along with the set of feasible configurations and set of feasible patterns denoted by $\mathcal{C}$ and $\mathcal{P}$, respectively. We now generate a MILP solution from $P'$. 

Do the following operation for each bin $b\in \bns$ in the packing $P'$. We define a configuration $c(b)$ corresponding to $b$. Identify the total size of parts of small items assigned to $b$, round it down to the next integer multiple of $\eps^2$, and store the integer multiple as $\gamma_c$ of $c(b)$. Now the remaining parts in $b$ are from large items. For each $r\in[\gs'/\eps^2]$, identify the number of different parts (split from large items) of size $r\eps^2$ assigned to bin $b$ and store it as $\delta_{cr}$ corresponding to component $r+1$ in $c(b)$. Let $\mathcal{C}'$ be the multi-set of configurations generated from all the bins. Since the total size of small items in each bin was rounded down, the size of any configuration is not more than the total size of the items in the corresponding bin in $P'$, i.e., at most $\gs'$.

Next, we generate the patterns. For each large item $i\in \mathcal{L}$, we generate a pattern $p(i)$ corresponding to $i$. If the size of $i$ is $\alpha\eps^2$, store the value of $\alpha$ as $\alpha_p$ for the pattern $p(i)$ corresponding to $i$. For each $r\in[\gs'/\eps^2]$, identify the number of parts of size $r \eps^2$ of $i$ created when we consider the splitting of $i$ (among all bins). This value is stored as the $\beta_{pr}$ corresponding to component $r+1$ in $p(i)$. Let $\mathcal{P}'$ be the multi-set of patterns generated from all the large items.  By definition, every such pattern is a feasible pattern.

Next, we generate the MILP solution. The variable $y_c, \forall c\in \mathcal{C}$ is set to the number of copies of $c$ in $\mathcal{C}'$. $z_p, \forall p\in \mathcal{P}$ is the number of copies of $p$ in $\mathcal{P}'$. From the packing, we know which parts of small items are assigned to which bins. If a portion of size $\tau S_i, \tau\in[0,1]$ of $i\in \mathcal{S}$ is assigned to a bin, consequently a configuration $c$, then by summing the $\tau$ values over all bins with configuration $c$, we get the $x_{ic}$ value as this sum. Now, we have $\textbf{x}$. The values of $\textbf{v}$ are defined using the $\textbf{z}$ values based on constraint \eqref{par_cou} of the MILP. Thus, we have the MILP solution $(\textbf{v},\textbf{x},\textbf{y},\textbf{z})$. We still need to prove that it satisfies the constraints of the MILP. 

Lastly we check the feasibility of the generated MILP solution. From the definition of $\textbf{y}$ Constraint \eqref{con_bin} is satisfied. Constraint \eqref{sma_con} is satisfied since all the small items are packed completely in $P'$ (some are packed fractionally since items can be split) and by definition of $\textbf{x}$. Note that $v_r = \sum_{p\in \mathcal{P}}\beta_{pr}z_p$ by the assignment of values of the variables $\textbf{v}$. Thus, Constraint \eqref{par_cou} is satisfied. Furthermore, since all the large items were packed (fractionally), all the parts cut from large items were packed among all the bins.  That is, $v_r = \sum_{c\in \mathcal{C}}\delta_{cr}y_c$. Thus, Constraint \eqref{con_par} is satisfied from the definition of configurations. Since all the large items were assigned a pattern, we have that the number of patterns chosen with first coordinate equal to $\ell/\eps^2$ is equal to the number of large items of size $\ell$, and thus Constraint \eqref{pat_lar} is satisfied. Because of the fact that $P'$ was a feasible nice packing and the number of parts of different items assigned to a bin is at most $k$, we have $\sum_{i\in \mathcal{S}}x_{ic} + y_c\sum_{r\in[\gs'/\eps^2]}\delta_{cr} \leq y_ck$ for all configurations $c\in \mathcal{C}$. Thus, Constraint \eqref{card} is satisfied. The total size of parts of small item in the collection of bins $B_c$ whose corresponding configuration is $c$ is $\sum_{i\in \mathcal{S}} x_{ic}S_i$. Since each bin in $B_c$ has a total size of small items in the interval $[\gamma_c \eps^2, (\gamma_c+1) \eps^2)$, so over the $y_c$ bins in $B_c$ we have a total size of small items that is not larger than $y_c(\gamma_c+1)\eps^2$, and Constraint \eqref{tot_sma} is satisfied. Constraint \eqref{non_v}, \eqref{non_x}, \eqref{non_y}, and \eqref{non_z} are satisfied trivially based on the definition of the $\textbf{v}$, $\textbf{x}$, $\textbf{y}$, and $\textbf{z}$.  
\end{proof}

\section{Converting the MILP solution into the output of the scheme}\label{milp_convert}

In order to present our rounding method of taking a solution to the MILP and returning a feasible solution to $I'$, we need a method to handle the assignment of the small items. We will apply a method established in \cite{chen2016} originally for the cardinality constraint scheduling problem.  Specifically the analysis of a procedure referred to as Best-Fit. Next, we state the lemma analyzing the Best-Fit solution for packing the small items satisfying cardinality constraint \cite[Lemma~2]{chen2016}. We apply this lemma in the proof of the following theorem.

\begin{lemma}\label{best_fit}
If there is a feasible fractional solution to the following feasibility linear program in the variables $\textbf{x}$ 
	\begin{align*}
		\sum_{i\in \mathcal{S}} S_ix_{ib} \leq t_b&,\ b\in \bns\\
		\sum_{i\in \mathcal{S}} x_{ib} \leq c_b&,\ b\in \bns\\
		\sum_{b\in \bns} x_{ib} = 1&,\ i\in \mathcal{S}\\
		0\leq x_{ib} \leq 1&, \ i\in \mathcal{S},\ \ b\in \bns
	\end{align*}
then an integer solution satisfying:
	\begin{align*}
		\sum_{i\in \mathcal{S}} S_ix_{ib} \leq t_b + S_{\max}&,\ b\in \bns\\
		\sum_{i\in \mathcal{S}} x_{ib} \leq c_b&,\ b\in \bns\\
		\sum_{b\in \bns} x_{ib} = 1&,\ i\in \mathcal{S}\\
		0\leq x_{ib} \leq 1&, \ i\in \mathcal{S},\ \ b\in \bns
	\end{align*}
could be obtained in $O(n\log n)$ time, where $S_{\max}=\max_{i\in \mathcal{S}}\{S_i\}\leq\eps$. 
\end{lemma}

\begin{theorem}\label{sol2pack}
If there exists a solution to the MILP for the rounded instance $I'$ with the given guessed value $\gs'$, then there exists a feasible packing to instance $I'$ of cost at most $(1+2\eps)\gs'$ that can be found in polynomial time.
\end{theorem}
\begin{proof}
Let the MILP solution be $(\textbf{v},\textbf{x},\textbf{y},\textbf{z})$, $\mathcal{C}$ the set of feasible configurations, and $\mathcal{P}$ the set of feasible patterns. We convert the MILP solution to a feasible packing.

For every $p\in \mathcal{P}$, assign $z_p$ large items, of size $\alpha_p \eps^2$, to pattern $p$. By Constraint \eqref{pat_lar}, we have enough patterns for all large items. Now we need to assign the parts given by $\textbf{v}$ and small items to bins. To do this for every $c\in \mathcal{C}$, we assign configuration $c$ to $y_c$ bins. Thus, all bins are assigned a configuration since Constraint \eqref{con_bin} is satisfied. 

Now we have information on how many parts of each size from large items we need to assign to each bin and a fractional assignment of small items to configurations. This information is used in the following way.  Assign $\delta_{cr}$ different parts of size $r\eps^2$ to each bin assigned configuration $c$ for $r\in[\gs'/\eps^2]$ in some order of $r$. All parts of size $r \eps^2, r\in[\gs'/\eps^2]$ are indeed packed in this way since Constraint \eqref{con_par} is satisfied. Next we increase the bin size by $\eps^2$ to pack the small items and this increase is needed since the total size of small items were rounded down in the configuration. From Constraint \eqref{tot_sma} this is enough space to pack all the small items fractionally as per $\textbf{x}$ as follows. For each $i\in \mathcal{S}$ and $c\in {\mathcal C}$, pack $x_{ic}/y_c$ fraction of small item $i$ to every bin $b$ assigned configuration $c$. Let $x_{ib},\forall i\in \mathcal{S}, \forall b\in \bns$ be the fractional packing of small items to bins. All the small items are packed since $\textbf{x}$ satisfies Constraint \eqref{sma_con}. This fractional packing may not be a feasible packing for \dsbp{} due to the cardinality constraint. 

Next, we modify this fractional packing of small items to an integral packing to get a feasible packing leaving the packing of the parts of large items unaltered. 
From Constraint \eqref{card}, we have that $\sum_{i\in \mathcal{S}}x_{ic}/y_c \leq k - \sum_{r\in[\gs'/\eps^2]}\delta_{cr}$. Thus, the number of small items that can be assigned integrally to a bin is at most $c_b = k -\sum_{r\in[\gs'/\eps^2]}\delta_{cr}$ where $c$ is the configuration assigned to bin $b$. From Constraint \eqref{tot_sma}, we have that the total size of small items that can be assigned to bin $b$ is at most $t_b = (\gamma_c+1)\eps^2$ where $c$ is the configuration assigned to bin $b$. $S_{\max}$ is the maximum size of a small item, i.e., at most $\eps$. We apply Lemma \ref{best_fit} to get the integral packing of small items denoted by $\textbf{x}'$ where $\textbf{x}$ is a feasible solution to the first feasibility linear program. From Lemma \ref{best_fit}, $\sum_{i\in \mathcal{S}}x'_{ib} \leq c_b = k -\sum_{r\in[\gs'/\eps^2]}\delta_{cr}$ where $c$ is the configuration assigned to bin $b$. Thus, we get $\sum_{i\in \mathcal{S}}x'_{ib} + \sum_{r\in[\gs'/\eps^2]}\delta_{cr} \leq k$  where $c$ is the configuration assigned to bin $b$. Thus, $\textbf{x}'$ along with the packing of parts of large items (as before) satisfy the cardinality constraint in each bin. The use of Lemma \ref{best_fit} increases the bin size by $S_{\max} \leq \eps$. Thus, the final bin size is at most $\gs'+\eps^2+\eps \leq(1+2\eps)\gs'$.
\end{proof}

We conclude that the following theorem is established.

\begin{theorem}\label{theorem_3}
	Problem \dsbp{} admits an EPTAS.
\end{theorem}


\end{document}